\documentclass[journal,twoside,web]{ieeecolor}
\usepackage{tuffc}
\usepackage{cite}
\usepackage{amsmath,amssymb,amsfonts}
\usepackage{algorithmic}
\usepackage{subcaption}
\usepackage{graphicx}
\usepackage{algorithm,algorithmic}
\usepackage{nicematrix}
\usepackage{hyperref}
\hypersetup{hidelinks=true}
\usepackage{multirow}
\usepackage{textcomp}
\usepackage{wrapfig,colortbl}
\definecolor{abstractbg}{rgb}{1,0.969,0.914}
\def\BibTeX{{\rm B\kern-.05em{\sc i\kern-.025em b}\kern-.08em
    T\kern-.1667em\lower.7ex\hbox{E}\kern-.125emX}}
\newtheorem{thm}{Theorem}[section]

\newtheorem{rem}{Remark}
\begin{document}
\title{Optimal GNSS Time Tracking for Long-term Stable Time Realisation in Synchronised Atomic Clocks}
\author{Maitreyee Dutta, Jiayu Chen, Masakazu Koike, Yuichiro Yano, Yuko Hanado, Takayuki Ishizaki
\thanks{This work is supported by the Ministry of Internal Affairs and Communications (MIC) under its “Research and Development for Expansion of Radio Resources (JPJ000254)” program.}
\thanks{Maitreyee Dutta, Jiayu Chen and Takayuki Ishizaki are with Department of Systems and Control Engineering, Institute of Science Tokyo, 2-12-1, Ookayama, Meguro, Tokyo, 152-8552, Japan.}
\thanks{Masakazu Koike is with Tokyo University of Marine Science and Technology, 4-5-7, Kounan, Minato, Tokyo, 108-0075, Japan.} \thanks{Yuichiro Yano and Yuko Hanado are with National Institute of Information and Communications Technology, 4-2-1,
Nukui-Kitamachi, Koganei, Tokyo, 184-8795 , Japan.}}
\IEEEtitleabstractindextext{%
\fcolorbox{abstractbg}{abstractbg}{%
\begin{minipage}{\textwidth}\rightskip2em\leftskip\rightskip\bigskip
\begin{wrapfigure}[11]{r}{3in}%
\hspace{-4.5pc}
\includegraphics[width=1\linewidth]{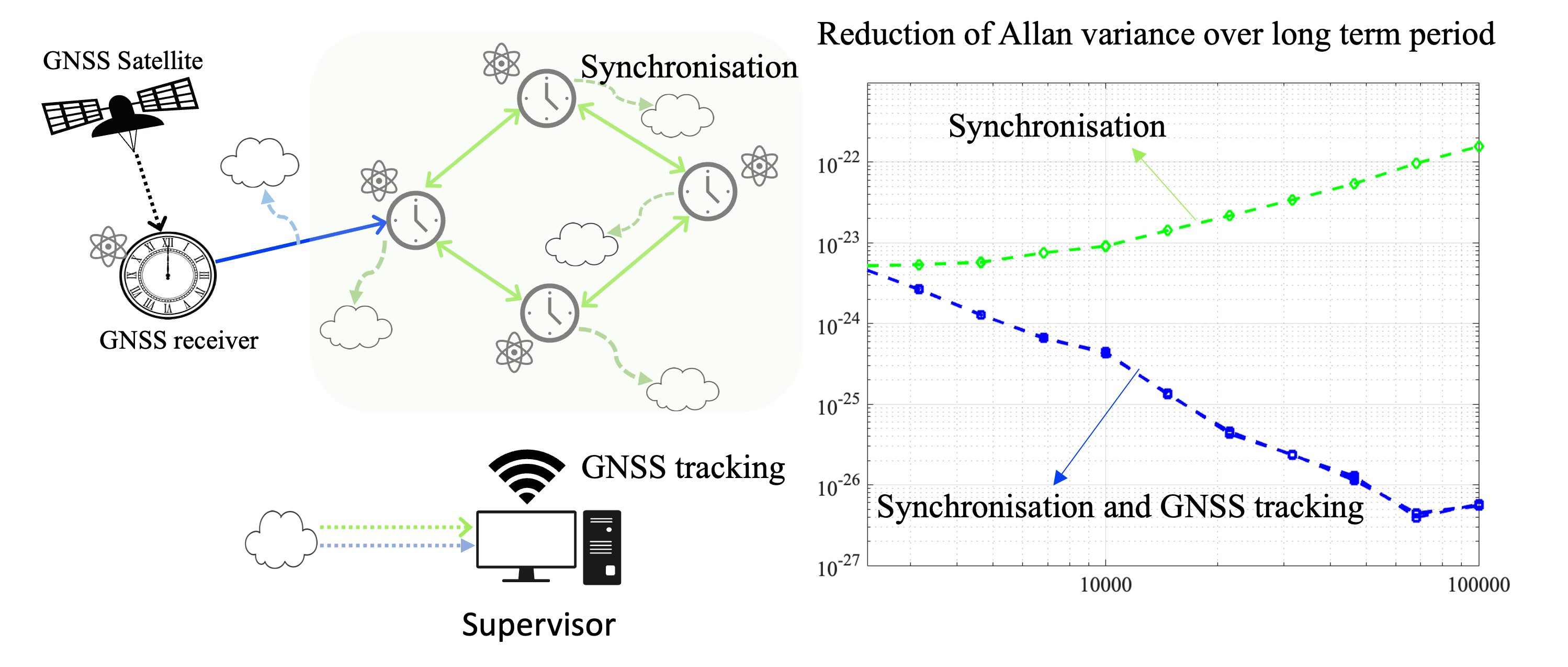}
\end{wrapfigure}%
\begin{abstract}
In this manuscript, we propose a novel optimal Global Navigation Satellite System (GNSS) time tracking algorithm to collectively steer an ensemble consisting of synchronising miniature atomic clocks towards standard GNSS time.
The synchronising miniature atomic clocks generate a common synchronised time which has good short term performance but its accuracy and precision, which is measured by Allan variance, deteriorates in the long run. 
So, a supervisor designs and periodically broadcasts the proposed GNSS time tracking control to the ensemble miniature atomic clocks that steer the average of ensemble towards the average of GNSS receivers, which are receivers of GNSS time.
The tracking control is constructed using a Kalman filter estimation process that estimates the difference in average of GNSS receivers and average of ensemble clocks by using relative clock readings between GNSS receivers and their adjacent ensemble clock.
Under the influence of the periodically received tracking control, the stabilised ensemble clocks have better long term accuracy and precision over long averaging periods. 
Since the tracking control is designed to solely influence the average of the ensemble, the tracking process does not interfere with the synchronisation process and vice versa.
The feedback matrix associated with the tracking control is obtained from an optimisation problem that minimises steady-state Allan variance.
Numerical results are provided to show the efficacy of the proposed algorithm for enhancing long term performance.
\end{abstract}

\bigskip
\end{minipage}}}

\maketitle

\section{Introduction}
\label{sec:introduction}
\IEEEPARstart{S}{ince} the last century, the advent of complex technological developments in the field of telecommunications, broadcasting, financial transactions, scientific exploration, etc., has prompted the requirement of  accurate timekeeping by synchronised atomic frequency standards \cite{csptp}. 
An atomic frequency standard is a frequency standard whose reference is the resonant frequency of atoms such as rubidium and cesium.
An atomic clock is an atomic frequency standard that is operated continuously 
to provide any user with frequency as well as time\cite{reviewafs}.
Miniature Atomic Clock (MAC) such as Chip Scale Atomic Clock (CSAC) offer a balance of low power consumption, frequency stability, and compactness making them an ideal candidate for integration into various compact portable devices in dense IOT network\cite{kitching}. 
\begin{table*}[!t]
\arrayrulecolor{subsectioncolor}
\setlength{\arrayrulewidth}{1pt}
{\sffamily\bfseries\begin{tabular}{lp{6.75in}}\hline
\rowcolor{abstractbg}\multicolumn{2}{l}{\color{subsectioncolor}{\itshape
Highlights}{\Huge\strut}}\\
\rowcolor{abstractbg}$\bullet$ & We propose an optimal GNSS time tracking control for a set of synchronising MACs that periodically steers the average of all the synchronising MACs towards the average of MACs in the GNSS receiver, thus decoupling mutual synchronisation and GNSS time tracking dynamics.\\
\rowcolor{abstractbg}$\bullet${\large\strut} & Despite unobservability of average of any set of MACs, we show that a summation procedure can be used to procure the measurements required for the estimation process related to tracking control scheme that will enhance both accuracy and precision of the synchronised time over a long time period.\\
\rowcolor{abstractbg}$\bullet${\large\strut} &The proposed tracking algorithm improves on the established results related to pairwise steering of frequency standards and synchronisation of a set of atomic clocks that generate synchronised time with optimum short term performance.\\[2em]\hline
\end{tabular}}
\setlength{\arrayrulewidth}{0.4pt}
\arrayrulecolor{black}
\end{table*}


Some of the commercially prevalent clock time synchronisation protocols, such as the Network Time Protocol (NTP)\cite{davidmills} and IEEE-1588 Precision Time Protocol (PTP)\cite{ptp}, have a hierarchical stratum architecture such that the time server in one stratum works as a time reference to the servers in the lower stratum, but accuracy deteriorates with each stratum \cite{pbanerjee}. 
The server in the topmost stratum is aligned with GNSS time that is kept within nanoseconds of accuracy to national standard time\cite{timekalmanfilter}.
GNSS time is transmitted by GNSS satellites around the world and is received by GNSS receivers \cite{receiver}. While GNSS receivers provide a standard time reference to the local MACs, the GNSS time transmitted to the receivers are sometimes vulnerable to factors such as electromagnetic interferences, troposphere, ionosphere condition \cite{gnssdrawbacks}. Therefore, MACs should function autonomously in the absence of GNSS time references.

One solution is clock synchronisation where an ensemble of MACs generate a common synchronised time that coincides with the time generated by a virtual average clock which is known to have excellent frequency stability over a short-term period.
Such procedures have been proposed in centralised\cite{eem,galleani} and distributed form \cite{chen} to obtain optimal performance whenever the GNSS services are unavailable.
However, the performance of synchronised time of the MACs deteriorates in the long run, so it is necessary to periodically provide GNSS time as reference to enhance the precision of synchronised time. 

In short, the GNSS time tracking scheme designed for an ensemble of MACs should improve synchronised time's
\begin{enumerate}
\item \textbf{\textit{long-term accuracy,}}
\item \textbf{\textit{long-term precision.}}
\end{enumerate}
 
Precision and accuracy represent two key metrics for judging the performance of the generated time.
Frequency stability is a fundamental characteristic of an oscillator that describes how well the same frequency is maintained over a certain period of time\cite{lombardi}.
Allan Variance is a commonly used measure for frequency stability of atomic clocks\cite{allan}. 
Accuracy of the generated time indicates its deviation with respect to the \textit{ideal time}. 

From the perspective of systems and control theory, there are several works related to tracking of a virtual reference time by a single follower atomic clock. For instance, in \cite{lqg}-\cite{lqgtx}, a hydrogen maser clock is steered to UTC(USNO) and a reference signal modeled as pure white frequency noise, respectively. Several established results cover pairwise steering of clocks, for example, a composite clock is generated by first gradually steering a voltage controlled oscillator (VCO) to a hydrogen maser clock which is in turn eventually steered to a cesium clock in \cite{dpllcom}, a hydrogen maser clock is steered to a cesium clock in \cite{dpll, priyanka}, a hydrogen maser clock is steered to a rubidium clock in \cite{pidbp}, an oven controlled quartz oscillator is steered to rubidium clock in \cite{pp}. However, none of them can be directly applied to stabilise and enhance long term performance of a set of synchronising MACs since the established control architecture relies on pairwise clock steering and errors associated with simultaneous synchronisation of numerous MACs can interfere with stabilisation of the mean of all the MACs.

The proposed tracking algorithm demonstrates an effective application of Kalman filter estimation, graph theory, control over networked systems and $H_2$-optimisation.



The rest of the paper is organised as follows. Section \ref{sec:problemformulation} covers some basic concepts related to atomic clocks to formulate the problem statement. The solution for the GNSS time tracking problem is provided in Section \ref{sec:mainresults}. Simulation results for illustration purposes are provided in Section \ref{sec:numericalresults}. The paper ends with concluding remarks in Section \ref{sec:conclusion}.
 
\section{Problem formulation}\label{sec:problemformulation}
\subsection{Notations}\label{notations}
$\mathbb{R}, \mathbb{C}, \mathbb{Z}$ denotes a set of real, complex and integer numbers, respectively.
$\boldsymbol{1}_n$ denotes a $n-$ dimensional vector whose elements are equal to one.
$I_n$ is a $n-$ dimensional identity matrix. 
$I_{n|i}$ denotes the $i^{th}$ row of the identity matrix $I_n$. 
Any diagonal matrix $O$ with elements $o_1, o_2, \ldots, o_n$ is denoted as $O=\text{diag}(o_1,\ldots,o_n)$ or $O=\text{diag}(o_i)_{i\in\{1,\ldots,n\}}$. 
For a set $\mathcal{H}, |\mathcal{H}|$ denotes its cardinality. 
If $H$ is a square matrix, then $\lambda(H)$ denotes the set of eigenvalues of $H$, $\lambda_{\max}(H)$ denotes the maximum eigenvalue in the set $\lambda(H)$.
Ker$(H)$ and Im$(H)$ denotes the kernel and image of matrix $H$, respectively.
Trace of matrix $H$ is denoted by tr$(H)$.
A matrix $N=\begin{bmatrix}N_1^\top&N_2^\top&\cdots&N_n^\top\end{bmatrix}^\top$, where $N_i$ is the $i^{th}$ row of $N$, can be denoted as $N= \text{col} (N_i)_{i\in \{1,\ldots,n\}}.$
For any matrices $C=[c_{ij}]\in \mathbb{R}^{p\times q}, D=[d_{ij}]\in \mathbb{R}^{s\times t}$, the matrix $C\otimes D\in\mathbb{R}^{ps\times qt}$ denotes the Kronecker product of the two matrices. For any vector $d\in \mathbb{R}^m, \|d\|_2\in \mathbb{R}$ represents its Euclidean norm. 
For any random variable $X$, $\mathbb{E}[X]$ denotes the expectation of $X$.

\subsection{Second order model of an atomic clock}\label{modelsingleclock}
In this subsection, we briefly recall the discrete-time second order atomic clock model \cite{higherallanvariance,tutorial}.

Consider a discrete time sequence $t_k=k\tau, k\in \mathbb{Z}, \tau>0,$ where the time period between two adjacent time instants is a constant, $(t_{k}-t_{k-1})=\tau \; \forall \,k\in \mathbb{Z}$. Here, $\tau$ is called the sampling period of the time sequence.
The discrete-time two state deviation model of atomic clocks is given by\cite{tutorial},
\begin{align*}
    \begin{bmatrix}x_1[k+1] \\ x_2[k+1]\end{bmatrix}=\begin{bmatrix}
        1&\tau\\0&1
    \end{bmatrix}\begin{bmatrix}x_1[k] \\ x_2[k]\end{bmatrix}+\begin{bmatrix}\tau\\1\end{bmatrix} u[k]+ v[k],
    \end{align*}
where $x_1[k]$ is the clock reading deviation or phase deviation from ideal time, $x_2[k]$ is the frequency deviation from ideal time, $u[k]$ is the control input to be designed and $v[k]=\begin{bmatrix}
    v_1[k]&v_2[k]
\end{bmatrix}^\top$ is the process Gaussian noise such that 
\begin{align*}
    \mathbb{E}[v[k]]=0,\quad \mathbb{E}[v[k] v^\top[k]]=Q_s,
\end{align*}
where $Q_s$ is the covariance matrix given by
\begin{align}
    Q_s=\begin{bmatrix}
      \tau\sigma_1^2+ \frac{\tau^3}{3}\sigma_2^2& \frac{\tau^2}{2} \sigma_2^2\\\frac{\tau^2}{2} \sigma_2^2& \tau \sigma_2^2
    \end{bmatrix},\label{Qs}
\end{align}
where $\sigma_1^2$ is the variance of white frequency noise and $\sigma_2^2$ is the variance of random-walk frequency noise.

For brevity, we will write the above discrete stochastic difference equation as 
\begin{subequations}
\begin{align}
& x[k+1]=A x[k]+ B u[k]+ v[k],\label{icmodel}\\
& d[k]=C x[k],\label{icmea}  
\end{align}
\end{subequations}

where $x[k]=\begin{bmatrix}x_1[k]&x_2[k]\end{bmatrix}^\top,$ $A\triangleq \begin{bmatrix}
        1&\tau\\0&1
    \end{bmatrix}, B\triangleq \begin{bmatrix}
        \tau\\1
    \end{bmatrix}, C\triangleq\begin{bmatrix}1 &0\end{bmatrix}$.
    
In the following subsection, we will discuss the networked system architecture that is primarily composed of an ensemble of several MACs, each of which can be modeled as per \eqref{icmodel}.


\subsection{Networked System Setup}\label{networkedsystemsetup} We consider an ensemble consisting of $n$ second-order MACs that are interacting as per an undirected network represented by graph $\mathcal{G}=(\mathcal{N},\mathcal{E})$ where $\mathcal{N}=\{1,\ldots,n\}$ is the index set of all clocks in the network and $\mathcal{E}\subseteq \mathcal{N}\times \mathcal{N}$ is the corresponding edge set consisting of all the edges in the network, as depicted in the light blue box of Figure \ref{fig:syn}- \ref{fig:broad}.
If $i^{th}$ clock receives any information from $j^{th}$ clock $\forall\; i\neq j$, then $(j,i)\in \mathcal{E}$. Furthermore, the $j^{th}$ clock is called a neighbour of the $i^{th}$ clock and $j\in \mathcal{N}_i$ where $\mathcal{N}_i$ represents the set of neighbours of the $i^{th}$ clock. Based on interconnectivity of the clocks,  we can construct an adjacency matrix $\mathcal{A}\triangleq [a_{ij}]\in \mathbb{R}^{n\times n}$ such that 
\begin{equation*}
   a_{ij} = \begin{cases}
         1, \;\text{if}\; j\in \mathcal{N}_i,\\
        0, \; \text{otherwise}.
    \end{cases}
\end{equation*}
For any undirected graph, the off-diagonal elements of the adjacency matrix satisfy $a_{ij}=a_{ji}.$
The degree matrix is $\mathcal{D}\triangleq \text{diag}(\sum_{j=1}^n a_{1j},\sum_{j=1}^n a_{2j},\ldots, \sum_{j=1}^n a_{nj})$.
Then, the Laplacian matrix can be defined as $\mathbb{R}^{n\times n}\ni\mathcal{L}\triangleq \mathcal{D}-\mathcal{A}.$

Relative clock reading between any two neighbouring clocks is measured using the Dual Mixture Time Difference (DMTD) technique\cite{pbanerjee}.
Due to physical restrictions, the absolute clock reading of any individual clock is not accessible. 
The relative clock reading measurement is called an edge measurement, and this is accessed by the ensemble MACs according to the underlying connection graph. 
We introduce the concept of the edge matrix of the $i^{th}$ MAC which is defined as $V_i\triangleq \text{col}(I_{n|j}-I_{n|i})_{j\in \mathcal{N}_i}\in \mathbb{R}^{|\mathcal{N}_i|\times n}$. 
Stacking the edge matrices of all $n$ MACs one on top of the other gives the edge matrix of the interconnected network as $V\triangleq \text{col}(V_i)_{i\in\mathcal{N}}\in\mathbb{R}^{|\mathcal{E}|\times n}$. For example, the edge matrix for the MACs in Fig. \ref{fig:example} is 
\begin{align*}
    V=\begin{bmatrix}
        V_1\\V_2\\V_3
    \end{bmatrix}=\begin{bmatrix}
        -1&1&0\\1&-1&0\\  0&-1&1\\ 0&1&-1  
    \end{bmatrix}.
\end{align*}

Secondly, there are $g$ dispersed GNSS receivers that have their own built-in MACs and are steered to GNSS time based on the signals received from GNSS satellites.
Naturally, it is assumed that $g < n$ and each receiver is adjacent to one of the ensemble MACs. 
Let $\mathcal{W}$ denote the subset of ensemble MACs that are adjacent to the receivers. 
Once again, the relative clock readings or edge measurements between the GNSS receivers and the adjacent clock are available by means of the DMTD technique.
\begin{figure}
\centering
\begin{subfigure}[b]{0.5\textwidth}
   \includegraphics[width=1\linewidth]{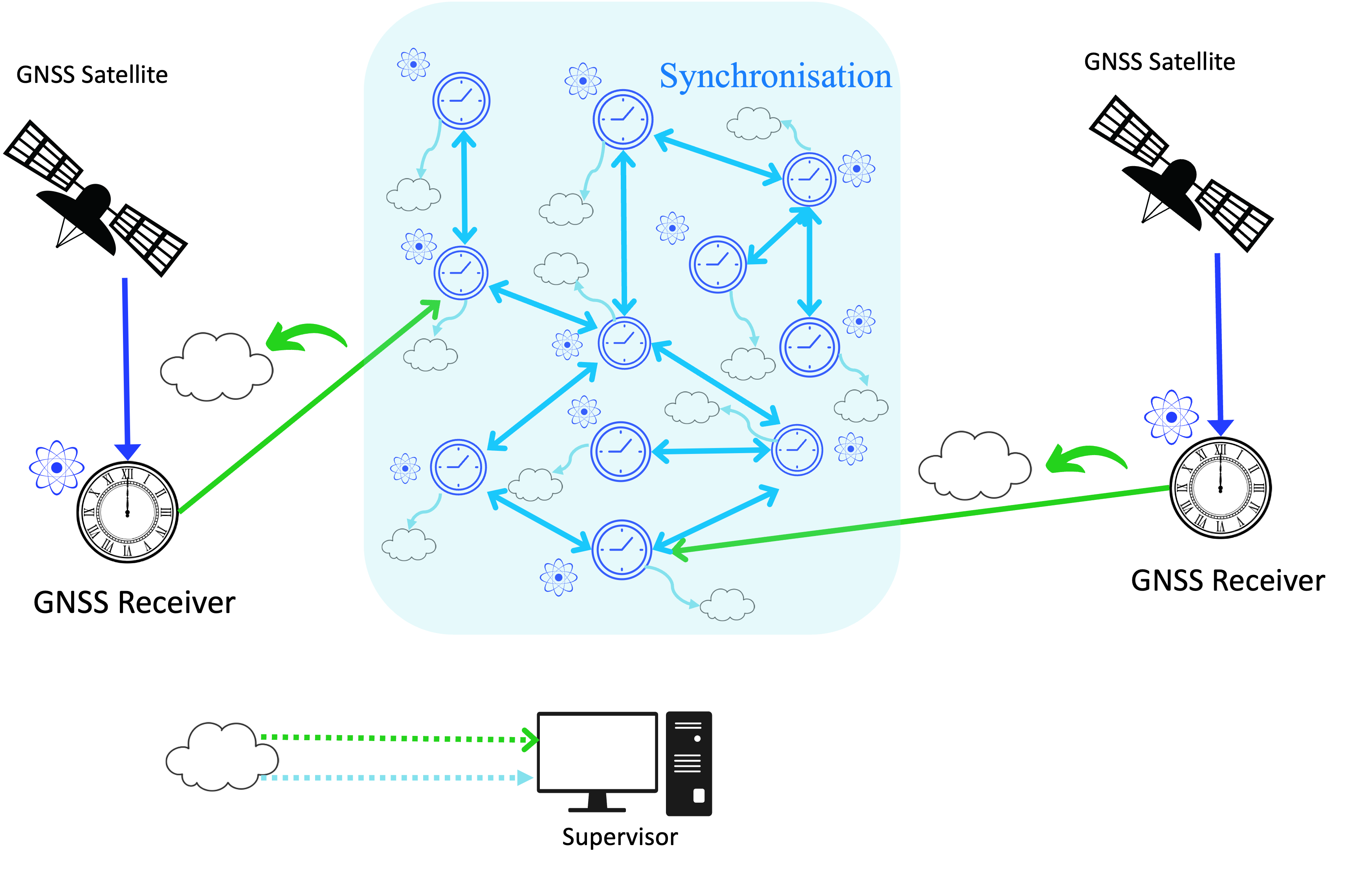}
   \caption{Persistent synchronisation of all MACs to the average of the ensemble.}
   \label{fig:syn} 
\end{subfigure}
\begin{subfigure}[b]{0.5\textwidth}
   \includegraphics[width=1\linewidth]{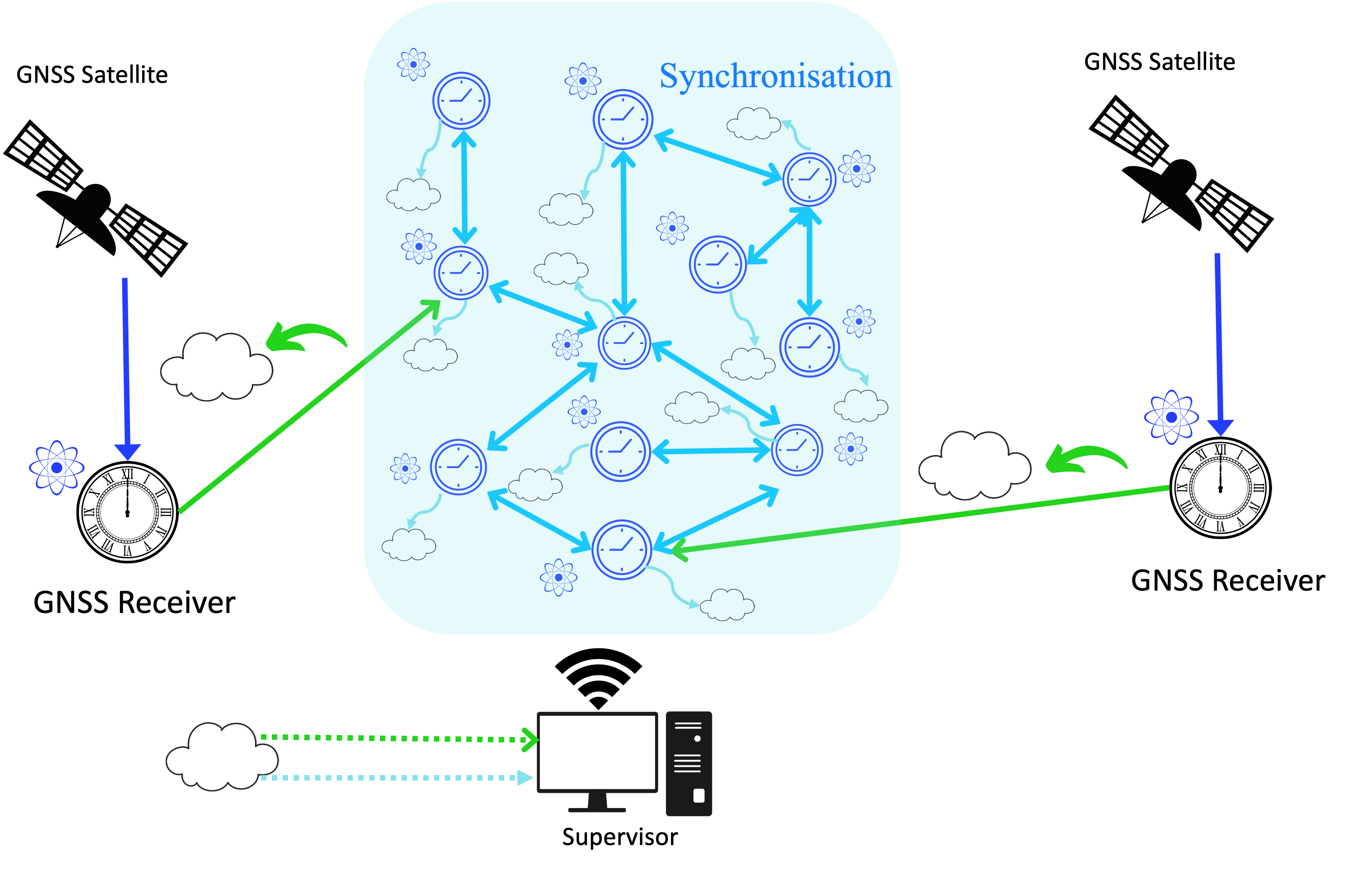}
   \caption{Periodic broadcasting of GNSS time tracking information by the supervisor to all the ensemble MACs.}
   \label{fig:broad}
\end{subfigure}
\caption[]{System setup for an ensemble consisting of MACs, GNSS receivers and supervisor.}
\label{fig:network}
\end{figure}
\begin{figure}
    \centering
    \includegraphics[width=1\linewidth]{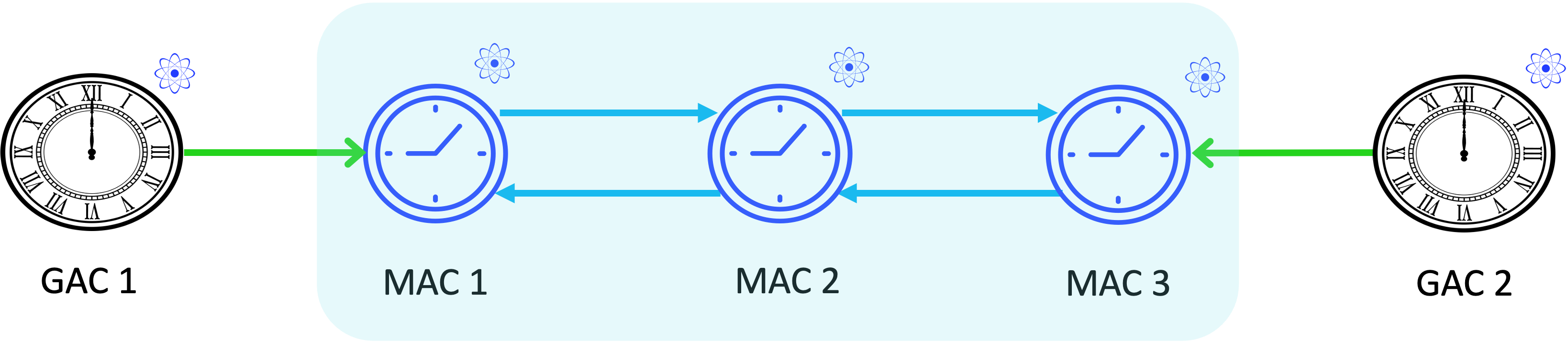}
    \caption{An example: Communication network of three MACs and two GACs}
    \label{fig:example}
\end{figure}
The corresponding edge matrix for such connections is given by
\begin{equation*}
    V_G=\begin{bmatrix}
      \text{col}(-I_{n|i})_{i\in \mathcal{W}}& I_{g}  
    \end{bmatrix}\in \mathbb{R}^{g\times (n+g)},
\end{equation*} where its $i^{th}$ row is $V_{G_i}=\begin{bmatrix}-I_{n|i}&I_{g|j}\end{bmatrix}$ such that $i^{th}$ MAC and the $j^{th}$ GNSS receiver are adjacent to each other. 
For example, in Fig. \ref{fig:example}, two GNSS receiver atomic clocks (GACs) are adjacent to MAC 1 and MAC 3. 
The edge matrix $V_G$ in this case is
\begin{align*}
    V_G=\begin{bNiceArray}{c c c : c c}
        -1 &0&0&1 &0\\0&0&-1&0&1
    \end{bNiceArray}.
\end{align*}

Finally, there exists a single supervisor that collects edge measurements between the GNSS receivers and their adjacent ensemble clocks, and all local estimates of relative phase and frequency deviation between neighbouring ensemble clocks to calculate the GNSS time tracking control which is to be broadcasted periodicially over a long period to the ensemble of $n$ MACs as shown in Figure \ref{fig:broad}.
The network configuration detailed above is shown in Figure \ref{fig:network}.

In the following subsection, we will cover the compact stochastic state-space model of $n$ cesium MACs and $g$ GNSS receiver MACs based on the second order model for a single atomic clock introduced in Subsection \ref{modelsingleclock}.

\subsection{Compact deviation model for a set of atomic clocks}
\subsubsection{Compact model for an ensemble of miniature atomic clocks}\label{compactcluster}

In the ensemble of $n$ MACs, any $i^{th}$ atomic clock model is $x_i[k+1]=A x_i[k]+ B u_i[k]+ v_i[k]$, where $i\in \mathcal{N}$.
Then, stacking all the $n$ models gives the compact form,
    \begin{align}
    x[k+1]= (I_n\otimes A) x[k]+(I_n\otimes B) u[k]+ v[k],\label{compactx}
\end{align}
where $x[k]\triangleq \begin{bmatrix}x_1^\top[k],\ldots,x_n^\top[k]\end{bmatrix}^\top\in \mathbb{R}^{2n},$ $u[k]\triangleq \begin{bmatrix}
    u_1[k],\ldots,u_n[k]
\end{bmatrix}^\top\in \mathbb{R}^n,$ $v[k]\triangleq \begin{bmatrix}v_1^\top[k],\ldots,v_n^\top[k]\end{bmatrix}^\top\in \mathbb{R}^{2n}.$
The compact system noise $v[k]$ is a stationary Gaussian process such that 
\begin{align*}
    \mathbb{E}[v[k]]=0, \; \mathbb{E}[v[k] v^\top[k]]=Q=\text{diag}(Q_i)_{i\in\mathcal{N}},
\end{align*}
where $Q_i$ is the covariance matrix akin to \eqref{Qs} and is parametrised by standard deviation $\sigma_{1i}$ and $\sigma_{2i}$ of clock $i$. 
Based on the edge matrix $V$ for the ensemble MACs introduced in Subsection \ref{networkedsystemsetup}, we introduce a variable called edge state $\xi[k]=(V\otimes I_2) x[k]$ which, represents the relative phase and frequency deviation between all neighbouring MACs in the network.
Then, $\xi_i[k]=(V_i\otimes I_2) x[k]= \text{col}( \xi_{ij}[k])_{j\in\mathcal{N}_i}$ is the edge state associated with the $i^{th}$ clock, where $\xi_{ij}[k]$ is the phase and frequency deviation of the $j^{th}$ clock with respect to the $i^{th}$ clock  and $\xi[k]=\text{col}(\xi_{i}[k])_{i\in \mathcal{N}}=\text{col}(\text{col}(\xi_{ij}[k])_{j \in \mathcal{N}_i})_{i\in \mathcal{N}}.$ 
Now, the $i^{th}$ clock in the ensemble has access to edge measurement $y_i[k]=(V_i\otimes C) x[k]+ w_i[k],$ which can be rewritten in terms of edge state $\xi_i[k]$ as $(I_{|\mathcal{N}_i|}\otimes C) \xi_i[k]+ w_i[k].$ 
Here, $w_i[k]$ is the measurement noise such that $\mathbb{E}[w_i[k]]=0, \mathbb{E}[w_i[k] w_i^\top[k]]=R_i\in \mathbb{R}^{|\mathcal{N}_i|\times |\mathcal{N}_i|}$.
The $i^{th}$ clock uses $y_i[k]$ to estimate edge state $\xi_i[k]$.

Now, all the edge measurements in the network can be stacked to get the compact edge measurement, 
\begin{align}
    y[k] =(I_{|\mathcal{E}|}\otimes C)\xi[k]+ w[k],
\end{align}
where $y[k]\triangleq \text{col}(y_i[k])_{i\in \mathcal{N}}$ and $w[k] \triangleq \text{col}(w_i[k])_{i\in \mathcal{N}}$ satisfies $\mathbb{E}[w[k]]=0, \mathbb{E}[w[k] w^\top[k]]= R= \text{diag}(R_i)_{i\in \mathcal{N}}\in \mathbb{R}^{|\mathcal{E}|\times |\mathcal{E}|}.$  

\subsubsection{Two state model for clocks in the GNSS receivers}
The MACs embedded in the GNSS receivers are synced to GNSS time before hand. 
However, because of broadcasted ephemeris inaccuracy, the receiver clocks display a constant offset in the range of the nanoscale with respect to the GNSS time.

So, the compact model of the clocks embedded in the GNSS receivers is given by
\begin{align}
    X[k+1]=(I_g\otimes A) X[k]+ v_G[k],\label{compactX}
\end{align}
where $X[k]=\begin{bmatrix}X_1^\top[k],\ldots,X_g^\top[k]\end{bmatrix}^\top$ such that $\mathbb{E}[X_i[0]]=\begin{bmatrix}
    \Theta_i\\0
\end{bmatrix}, \Theta_i\in \mathbb{R}$ and $v_G[k]$ is the system noise such that 
\begin{align}
 \mathbb{E}[v_G[k]]=0,\quad \mathbb{E}[v_G[k]v_G^\top[k]]=Q_G=\text{diag}(Q_{G_i})_{i\in \{1,\ldots,g\}},   \nonumber
\end{align}
where $Q_{G_i}$ is the covariance matrix akin to \eqref{Qs} and is parametrised by standard deviation $\sigma_{G_{1i}}$ and $\sigma_{G_{2i}}$. 
Since all the receiver clocks follows the GNSS time, despite some small offset, the system noise standard deviation $\sigma_{G_{1i}}, \sigma_{G_{2i}}$ is significantly lower than the system noise standard deviation $\sigma_{1i}, \sigma_{2i}$ associated with the ensemble MACs.
Based on the edge matrix $V_G$ introduced in Subsection \ref{networkedsystemsetup}, we can now describe its associated edge state $\xi_G[k]=(V_G\otimes I_2)\begin{bmatrix}x[k]\\X[k]\end{bmatrix}$, which describes the relative phase and frequency deviation between the GNSS receiver clocks and their adjacent ensemble clocks.
Then the related output measurement process $Y[k]=(V_G\otimes C) \begin{bmatrix}x[k]\\X[k]\end{bmatrix}+ w_G[k]$ can be written in terms of edge state as 
$(I_g\otimes C)\xi_G[k] + w_G[k].$ Here, $w_G[k]$ is the measurement noise that satisfies $\mathbb{E}[w_G[k]]=0, \mathbb{E}[w_G[k] w_G^\top[k]]=R_G \in \mathbb{R}^{g\times g}.$
The supervisor accesses the edge measurement $Y[k]$ to estimate a variable called \textit{tracking error} which will be described in the upcoming subsections.
Due to availability of various noisy edge measurements, Kalman filter can be employed to estimate the associated edge state.

\subsection{Frequency stability: Allan variance(AVAR)}\label{sec:avar}
Let us consider the individual second-order clock model as per \eqref{icmodel}-\eqref{icmea} when it is unforced, \textit{i.e.} $u[k]=0$. 
Then, its AVAR\cite{allan} is
\begin{align}
    \sigma_A^2(\tau) = \mathbb{E}\bigg[\frac{(\Delta_1^2 d[k])^2}{2\tau^2}\bigg],
\end{align}
where $\Delta_1^2 d[k]$ is the second-order difference of $d[k]$. 
In fact, the AVAR of an unforced or free running atomic clock turns out to be a function of $\tau$ only and is expressed as \cite{zucca}
\begin{align}
    \sigma_A^2(\tau) = \frac{1}{\tau} \sigma_1^2+\frac{\tau}{3}\sigma_2^2.
\end{align}
However, when the clocks are controlled, \textit{i.e.} $u[k]\neq 0$, AVAR is computed statistically based on output sequence $\{d[0], d[1], \ldots, d[T]\}$ which is measured over the sampling period $\tau$.
So, the statistical estimation of AVAR is given by \cite{stav}
\begin{align}
    \sigma_A^s(w;\{d[k]\})=\frac{1}{(T-2w)} \sum_{k=0}^{(T-2w-1)} \frac{(\Delta_w^2 d[k])^2}{2(w\tau)^2},
\end{align}
where $w$ is an interval. 
From the above expressions, it can intuitively be concluded that a relatively low AVAR indicates higher frequency stability in the timescale generated by the clock.  
As we are primarily dealing with a set of MACs, let us check the AVAR of the virtual clock which is the ensemble mean or average of a set of clocks. 
Consider the notation $\Gamma(\tau)\triangleq \tau \Sigma_1+\frac{\tau^3}{3}\Sigma_2$ where $\Sigma_1=\text{diag}(\sigma_{1i}^2)_{i\in\mathcal{N}}$ and $\Sigma_2=\text{diag}(\sigma_{2i}^2)_{i\in\mathcal{N}}$.
Then, the AVAR of the ensemble mean clock is given by\cite{eem}
\begin{align}
    \sigma_A^2(\tau)= \frac{1}{\tau^2}q^\top \Gamma(\tau) q,\label{ensembleavar}
\end{align}
where $q=(\boldsymbol{1}_n^\top\boldsymbol{1}_n)^{-1}\boldsymbol{1}_n=\frac{1}{n}\boldsymbol{1}_n,$ so it can be seen that $q^\top\boldsymbol{1}_n=1.$ 
\begin{table}[t!]
    \centering
    \begin{tabular}{c|c|c|c|c}
    &Scale& $i=1$ & $i=2$& $i=3$\\
    \hline &&&&\\
        $\sigma_{1i}^2$&$10^{-18}$ &0.0289&$7.84996\times 10^{-3}$ & $0.0149$ \\&&&&\\ 
    \hline &&&&\\
        $\sigma_{2i}^2$& $10^{-24}$ & $0.0227$ & $2.83\times 10^{-3}$ & $2.7889\times 10^{-4}$\\
        &&&&
    \end{tabular}
    \caption{Variances of the noises associated with three cesium clocks}
    \label{tab:variances of the noises}
\end{table}
\begin{figure}[t!]
    \centering
    \includegraphics[width=0.7\linewidth]{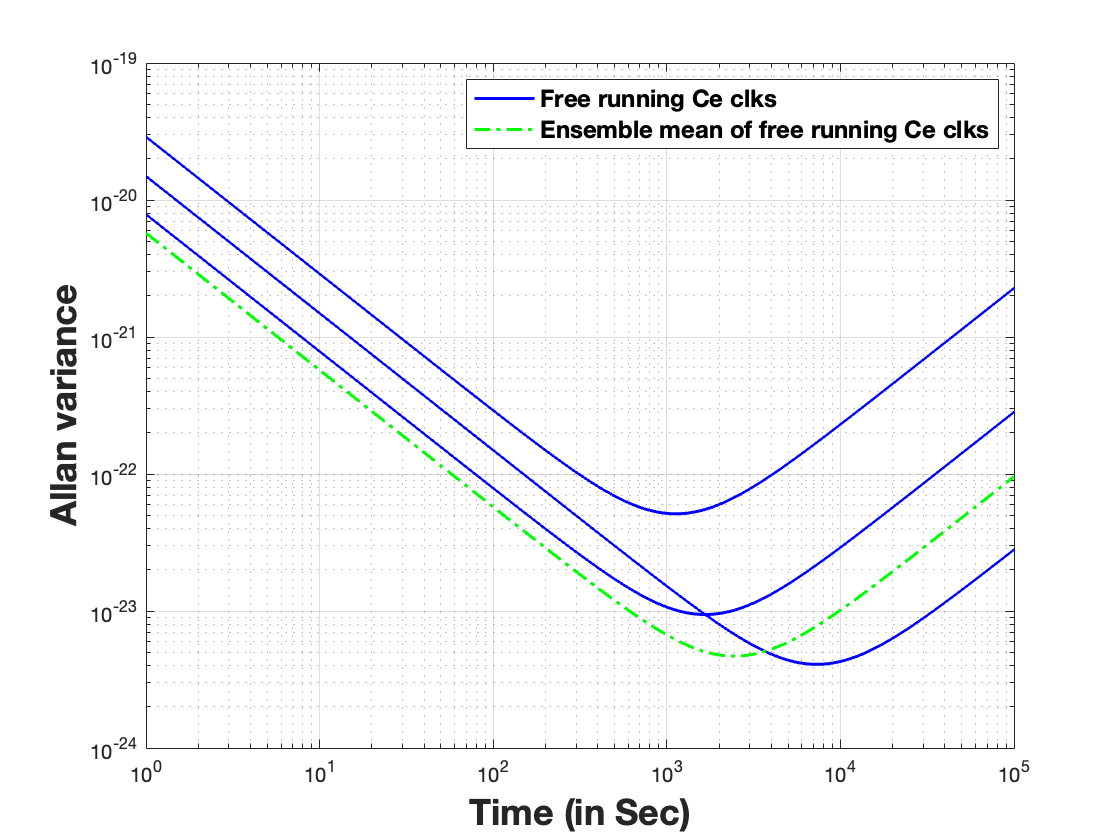}
    \caption{AVAR of three free running clocks and average of these three clocks.}
    \label{fig:avarfr}
\end{figure}
The average of free running cesium clocks has superior short-term performance compared to any of the individual clocks. 
This is because the AVAR of ensemble mean clock is lower than AVAR of any of the individual free running clocks. 
For illustration purposes, we have shown the AVAR of three free running cesium clocks whose white frequency noise and random-walk frequency noise variances are listed in Table \ref{tab:variances of the noises}. 
From Figure \ref{fig:avarfr}, it can be seen that AVAR of the ensemble mean (green plot) in short term is lower than AVAR of any of the individual clock (blue plots). 
However, the performance of ensemble mean AVAR deteriorates after 2000 Sec.

\subsection{Distributed synchronisation among ensemble clocks}
The $n-$ MACs estimate their incoming edges based on the noisy edge measurements available to them according to the underlying connected graph $\mathcal{G}$ to ensure distributivity in the control loop. 
The synchronising  control that is applied to any $i^{th}, i\in \mathcal{N}$ ensemble clock is\cite{chen},
\begin{align}
    u_i[k]=\frac{1}{2} F\sum_{j\in \mathcal{N}_i} (\hat{\xi}_{ij}^-[k]-\hat{\xi}_{ji}^-[k]), \label{usyni}
\end{align}
where $F\in \mathbb{R}^{1\times 2}$ is selected to satisfy the condition
\begin{align}
    \lambda(\Pi\otimes A&-\mathcal{L}\otimes BF)\in \mathbb{D},\label{fcondition}
\end{align}
where $\Pi=(I_n-\boldsymbol{1}_n q^\top)$ and $\mathbb{D}$ is the unit circle in the complex plane such that $\mathbb{D}\triangleq \{y\in\mathbb{C}: |y|<1\}$. In \eqref{usyni},
$\;\hat{\xi}_{ij}^-[k],\hat{\xi}_{ji}^-[k]$ are Kalman filter estimates of $\xi_{ij}[k]$ and $\xi_{ji}[k]$ based on measurements $y_i[k]$ and $y_j[k],$ respectively. The local Kalman filter gain used for estimation in any clock $i$ is given by $H_i^*= A_i P_iC_i^\top (C_iP_iC_i^\top +R_i)^{-1}$, $P_i$ being solution of
\begin{align}
    P_i=A_i P_i A_i^\top- A_i P_iC_i^\top (C_iP_iC_i^\top+R_i)^{-1} C_i P_iA_i^\top+\bar{Q}_i\nonumber
\end{align}
where $A_i=(I_{|\mathcal{N}_i|}\otimes A), C_i=(I_{|\mathcal{N}_i|}\otimes C), \bar{Q}_i= (V_i\otimes I_2)Q(V_i\otimes I_2)^\top$. 
As we have assumed that the underlying graph $\mathcal{G}$ is undirected, any $i^{th}$ clock can easily access the estimate $\sum_{j\in \mathcal{N}_i}\hat{\xi}_{ji}^-[k]$ without sacrificing distributivity in the control loop. It can be verified that on application of \eqref{usyni}, in steady state, each of the MACs satisfies $\mathbb{E}[x_i[k]]=\mathbb{E}[\bar{z}[k]],\;\forall\, i\in \mathcal{N}$ where $\bar{z}[k]= (q^\top\otimes I_2)x[k].$
So, due to continuous application of distributed synchronisation, all the ensemble MACs simultaneously converge to the virtual average clock that is the average of all MACs in the network which is known to have the best short-term performance as seen in Subsection \ref{sec:avar}.

Now, to improve the long-term accuracy and precision of the common synchronised time generated by distributedly synchronised MACs, a $H_2$-optimal GNSS time tracking control should be designed that can periodically regulate the ensemble mean $\bar{z}[k]$ by steering it towards the GNSS receiver clocks.
\subsection{Problem statement}\label{sec:problemstatement}
The primary objective of this paper is to construct a GNSS time tracking control for the ensemble MACs  that will improve the two performance metrics of the generated synchronised time as mentioned in Section \ref{sec:introduction}.
For the given networked system setup as shown in Figure \ref{fig:network}, improvement in the performance can be expressed mathematically as follows.

\textbf{Tracking of standard time} The aim is to steer the synchronised time towards GNSS time. This can be expressed mathematically as follows.
\begin{align}
    \underset{k\to \infty}{\lim}\mathbb{E}\Bigg[\bigg{\|}\bar{z}[k]-\frac{1}{g}\sum_{j=1}^gX_j[k]\bigg{\|}_2\Bigg]=0.\label{trackopti}
\end{align}
\begin{rem}
    The desired mathematical objective given in \eqref{trackopti} can be expanded and interpreted as follows.
    \begin{subequations}\label{trackobjectives}
\begin{equation}
  \underset{k\to \infty}{\lim}\mathbb{E}[\bar{z}[k]]=\underset{k\to \infty}{\lim}\mathbb{E}\bigg[\frac{1}{g}\sum_{j=1}^gX_j[k]\bigg]
\end{equation}
\begin{equation}
\underset{k\to \infty}{\lim}\mathbb{E}\Bigg[\bigg(\bar{z}[k]-\frac{1}{g}\sum_{j=1}^gX_j[k]\bigg)\bigg(\bar{z}[k]-\frac{1}{g}\sum_{j=1}^gX_j[k]\bigg)^\top\Bigg] \\=c,
\end{equation}
\end{subequations}
where $c$ is a finite value.
\end{rem}



In the following section, we will expand on GNSS time tracking procedure which will stabilise $\bar{z}[k]$ according to the above mathematical statements despite unobservability of $\bar{z}[k]$.

\section{Optimal GNSS time tracking control}\label{sec:mainresults}
The supervisor calculates the GNSS time tracking control and broadcasts it to all the ensemble MACs periodically over long period of time. 
For further development, we will use the basis vector $q_G=(\boldsymbol{1}_g^\top\boldsymbol{1}_g)^{-1} \boldsymbol{1}_g=\frac{1}{g}\boldsymbol{1}_g$ to represent the average vector of the MACs embedded in the GNSS receivers. 
So, the average of GNSS receiver clocks is denoted by $ \bar{Z}[k]\triangleq (q_G^\top\otimes I_2) X[k]$.
The primary objective is to minimise the deviation between $\bar{z}[k]$ and $\bar{Z}[k]$. 
The difference between the average subsystems is called tracking error, which is denoted by $\tilde{z}[k]\triangleq \bar{z}[k]-\bar{Z}[k]$, whose dynamics is
\begin{align}
  \tilde{z}[k+1]
  =& A \tilde{z}[k] + (q^\top\otimes B) u[k]+ (q^\top\otimes I_2) v[k]\nonumber\\&- (q_G^\top\otimes I_2) v_G[k].\label{ztilde}
\end{align}
We will use the shorthand notations $\bar{v}[k]\triangleq(q^\top\otimes I_2) v[k]$ and $\bar{v}_G[k]=(q_G^\top\otimes I_2) v_G[k]$ hereafter.
As apparent from the above equation, to influence the destination of synchronised clocks towards GNSS time, we need information about $\tilde{z}[k]$ or both $\bar{z}[k]$ and $\bar{Z}[k]$, which are not available in the measurement process, rather the edge states are.
The inaccessibility of tracking error or averages of the two subsystems can make the tracking procedure challenging, but in the following section, we will apply a simple summation procedure to the accessible measurement process to remove this encumbrance.
\subsection{Estimation of the tracking error}
Recall the output edge measurement process between the GNSS receiver and their adjacent clocks in the ensemble, this is accessed by the supervisor,
\begin{align}
     Y[k]
     = (I_g\otimes C) \xi_G[k]+ w_G[k].\label{Yk}
\end{align}

For further derivation, we will utilise a vector $q_A$ which will indicate adjacency of the ensemble clocks to the GNSS receiver and its elements are defined as,
\begin{equation}\label{qA}
    q_A(i)=\begin{cases}
        1,\, \text{if} \;i\in\mathcal{W},\\
        0, \text{otherwise}.
    \end{cases}
\end{equation}
The supervisor applies a summation operation to $Y[k]$ which is denoted by $\bar{Y}[k]\triangleq \boldsymbol{1}_g^\top Y[k]$. Using \eqref{Yk} and \eqref{qA}, $\bar{Y}[k]$ can be expanded and rewritten in terms of $\tilde{z}[k]$ as shown below,
\begin{align*}
   \bar{Y}[k]
   =-g\,C\,\tilde{z} [k] +((g \,q^\top -q_A^\top)\otimes C) x[k] +\boldsymbol{1}_g^\top w_G[k].
\end{align*}
As $ \text{Ker}(g\,q^\top-q_A^\top)=\text{Im}(\boldsymbol{1}_n),$ so we can rewrite $\bar{Y}[k]$ as 
\begin{align}
    \bar{Y}[k]
    =-g C \tilde{z} [k] - (q_A^\top\otimes C) z[k] +\boldsymbol{1}_g^\top w_G[k],\label{Ybar}
\end{align}
where $z[k]=(V^\dagger V\otimes I_2) x[k]=((I_n-\boldsymbol{1}_n q^\top)\otimes I_2) x[k]$ is the synchronisation error and $V^\dagger$ is the generalised inverse of $V$ such that $\text{Ker}(q^\top)= \text{Im}(V^\dagger)$.
Now, the coagulated edge measurement \eqref{Ybar} is a function of both tracking error $\tilde{z}[k]$ and synchronisation error $z[k]$. 
Even though, we are in a position to apply the standard Kalman estimation algorithm to estimate $\tilde{z}[k]$ based on measurement $\bar{Y}[k]$, the by-product term $(q_A^\top\otimes C) z[k] $ should be compensated. 
The estimation process used by the supervisor is summarised as 
\begin{align}
     \hat{\tilde{z}}^-[k+1]=&A\hat{\tilde{z}}^-[k]+ A H_{\tilde{z}}^* (\bar{Y}[k]+  gC\hat{\tilde{z}}^-[k]\nonumber\\+& (q_A^\top V^\dagger\otimes C) \hat{\xi}^-[k])+ (q^\top\otimes B) u[k],\label{ztildehat-}
\end{align}
where $ H_{\tilde{z}}^*$ is the Kalman gain obtained from
\begin{subequations}\label{pztilde}
    \begin{align}
      P_{\tilde{z}}^-=& AP_{\tilde{z}}^- A^\top+\bar{Q}_G\nonumber\\- A P_{\tilde{z}}^-&(-gC)^\top (g^2 C P_{\tilde{z}}^- C^\top +\bar{R}_G)^{-1}(-gC)P_{\tilde{z}}^- A^\top\\
       H_{\tilde{z}}^* = &P_{\tilde{z}}^- (-gC)^\top (g^2 C P_{\tilde{z}}^- C^\top +\bar{R}_G)^{-1},
    \end{align}
\end{subequations}
where $\bar{R}_G=(\boldsymbol{1}_g^\top\otimes I_2)R_G(\boldsymbol{1}_g^\top\otimes I_2)^\top$ and \\$\bar{Q}_G=
   (q^\top\otimes I_2)Q(q^\top\otimes I_2)^\top+(q_G^\top\otimes I_2) Q_G (q_G^\top \otimes I_2)^\top$.
In \eqref{ztildehat-}, the supervisor compensates the synchronisation error based by-product term that arises from the coagulated measurement $\bar{Y}[k]$ by utilising edge state estimates $\hat{\xi}^-[k]=\text{col}(\hat{\xi}_i^-[k])_{i\in \mathcal{N}}=\text{col}(\text{col}(\hat{\xi}_{ij}^-[k])_{j \in \mathcal{N}_i})_{i\in \mathcal{N}},$ of the MACs in the ensemble, which are collected as shown in Figure \ref{fig:broad}.

\subsection{Tracking control broadcasted by the supervisor}\label{trackingcontrol}
At the time of broadcast, the control input transmitted by the supervisor to the ensemble can be expressed compactly as 
\begin{align}
    u_{G}[k]= -(\textbf{1}_n\otimes F_B)  \hat{\tilde{z}}^-[k], \label{ug}
\end{align}
where $F_B$ is the feedback matrix assigned for effective implementation of tracking control. Let us denote the broadcasting period as $s$. So the overall control input that is applied to the MACs is
\begin{align}
    &u[k]=\begin{cases}
          u_{syn}[k]+u_G[k]\hspace{-0.35cm}&,\text{if} \, k+1 = ls, l\in \mathbb{Z}^+, \\
          u_{syn}[k]\hspace{-0.35cm}&,\text{otherwise},
    \end{cases}\label{ufinal}
\end{align}
where $u_{syn}[k]=\frac{1}{2} \big(\text{diag}(\boldsymbol{1}^\top_{|\mathcal{N}_i|})_{i\in\mathcal{N}}\otimes F\big)(\hat{\xi}^-[k]-\hat{\xi}_r^-[k])$ is the compact form of \eqref{usyni} and $\hat{\xi}_r^-[k]=\text{col}(\text{col}(\hat{\xi}_{ji}^-[k])_{j \in \mathcal{N}_i})_{i\in \mathcal{N}}$.

To conform with the control objective \eqref{trackopti}, the feedback matrix $F_B=\begin{bmatrix}f_{b1}& f_{b2}\end{bmatrix}$ is obtained by solving the following $H_2-$optimisation problem,
\begin{align}
        \underset{f_{b1}, f_{b2}}{\min}\sqrt{\mathbb{E}[\tilde{z}^\top[\infty]\tilde{z}[\infty]]} = \underset{f_{b1}, f_{b2}}{\min}\sqrt{\text{tr}(\tilde{C} \tilde{P} \tilde{C}^\top)}\label{fbopti}
    \end{align}
    subject to
$\frac{4}{(f_{b1}\tau s+2f_{b2})}>1,$ where $\tilde{P}$ solves $\tilde{A} \tilde{A}_0^{(s-1)} \tilde{P} \big(\tilde{A} \tilde{A}_0^{(s-1)}\big)^\top - \tilde{P}+\tilde{Q} =0,$
    \begin{align*}
\tilde{A}&=\begin{bmatrix}
        (A-BF_B)& BF_B &0\\
        0& A(I_2-H_{\tilde{z}}^* (-gC))&\tilde{a}_{23}\\
        0&0&\tilde{a}_{33}
\end{bmatrix},\nonumber\\
 \tilde{A}_0&=\begin{bmatrix}
        A& 0 &0\\
        0& A(I_2-H_{\tilde{z}}^* (-gC))&\tilde{a}_{23}\\
        0&0&\tilde{a}_{33}
\end{bmatrix},\nonumber\\
\tilde{a}_{23}&=A H_{\tilde{z}}^* (q_A^\top V^\dagger\otimes C),\nonumber\\
\tilde{a}_{33}&=(I_{|\mathcal{E}|}\otimes A) (I_{2|\mathcal{E}|}- H^*(I_{|\mathcal{E}|}\otimes C)),
\end{align*}
$H^*=\text{diag}(H_i^*)_{i\in \mathcal{N}}$, $\tilde{C}=\begin{bmatrix}I_2 &0_{2\times 2}&0_{2\times 2|\mathcal{E}|}\end{bmatrix}$, and $\tilde{Q}$ is a positive definite symmetric matrix.

The effect of periodically applied tracking control in the synchronised time is summarised in the following theorem.
\begin{thm}\label{thmfb}
    Consider a set of $n$ MACs \eqref{compactx} communicating by an undirected connected graph $\mathcal{G}=(\mathcal{N},\mathcal{E})$ and, synchronising as per \eqref{usyni} based on local edge state estimates.
    Then, $\underset{k\to \infty}{\lim} \mathbb{E}[\tilde{z}[k]]=0, \underset{k\to \infty}{\lim} \mathbb{E}[\tilde{z}[k]\tilde{z}^\top[k]]=C$, where $C$ is a constant matrix, is satisfied if and only if the supervisor constructs and periodically broadcasts the tracking control to the synchronising MACs as per \eqref{ug} based on estimation process \eqref{ztildehat-}-\eqref{pztilde} and optimisation problem in \eqref{fbopti}.
\end{thm}
\begin{proof}
For brevity, we will use notations $A_B(\gamma)\triangleq (A-B F_B)A^\gamma, \gamma \in \mathbb{R}^+$ and $\bar{v}_T[k]\triangleq\bar{v}[k]-\bar{v}_G[k]$. This implies $A_B(0)=(A-B F_B)$ and $A_B(\gamma)^0=I_2.$

 The effect on application of control input \eqref{ufinal} to difference equation of tracking error can be seen by deriving its closed loop equation,
    \begin{align}\label{ztildedyn}
       &\tilde{z}[k+1]=\nonumber\\
       &\begin{cases}
           A_B(0) \tilde{z}[k] + B F_B e_{\tilde{z}}^-[k]
  + \bar{v}_T[k]&\hspace{-0.35cm},\text{if} \;k+1= ls, l\in \mathbb{Z},\\ 
 A \tilde{z}[k] + \bar{v}_T[k]&\hspace{-0.35cm},\text{otherwise}.
       \end{cases} 
    \end{align}
Let us expand \eqref{ztildedyn} at the broadcast instant, 
\begin{align}
\tilde{z}[ls]=&A_B(s-1)^l  \tilde{z}[0]+\sum_{j=1}^l A_B(s-1)^{(l-j)}B F_B e_{\tilde{z}}^-[js-1] \nonumber\\&+\sum_{j=1}^{l}A_B(s-1)^{(l-j)} \bigg(\sum_{i=0}^{(s-2)} A_B(i)  \bar{v}_T[js-2-i]\nonumber\\&+ \bar{v}_T[js-1]\bigg),\label{ztildeexpansion} 
\end{align}
where $e_{\tilde{z}}^-[k]=\tilde{z}[k]-\hat{\tilde{z}}^-[k]$ is the error in estimation of tracking error. As the estimate $\hat{\tilde{z}}^-[k]$ is obtained from Kalman filter so $\underset{k\to \infty}{\lim}\mathbb{E}[e_{\tilde{z}}^-[k]]=0$. Moreover, the system noises satisfies $\mathbb{E}[v[k]]=0, \mathbb{E}[v_G[k]]=0,\; \forall\; k\in \mathbb{Z}$. Therefore, to stabilise $\tilde{z}[ls]$, it is apparent from \eqref{ztildeexpansion} that the coefficient matrix $A_B(s-1)$ should be Schur stable. In other words, $\lambda (A_B(s-1))\in \mathbb{D},$ $\mathbb{D}$ is the unit disk in complex plane such that $\mathbb{D}\triangleq \{y\in\mathbb{C}: |y|<1\}$. In the following paragraph, we will show that assigning $F_B$ as per \eqref{fbopti} indeed guarantees Schur stability of the coefficient matrix. 

The $H_2$-optimisation problem in \eqref{fbopti} is constructed using dynamics of $\tilde{z}[k],$ error in estimation of tracking error $\tilde{z}[k],$ $e_{\tilde{z}}^-[k]= \tilde{z}[k]-\hat{\tilde{z}}^-[k],$ and error in estimation of edge states $\xi[k]$,  $e^-[k]=\xi[k]-\hat{\xi}^-[k]$.  In fact, regardless of the time instant, the dynamics of $e_{\tilde{z}}^-[k]$ is 
\begin{align}
    e_{\tilde{z}}^-[k+1]=&A (I_2 - H_{\tilde{z}}^*(-gC))e_{\tilde{z}}^-[k] -A H_{\tilde{z}}^* w_G[k] \nonumber\\&+A H_{\tilde{z}}^*(q_A^\top V^\dagger\otimes C ) e^-[k]+ \bar{v}[k]- \bar{v}_G[k],\label{eztilde}
\end{align}
and the dynamics of $e^-[k]$ is 
    \begin{align}
         e^-[k+1]=&(I_{|\mathcal{E}|}\otimes A) (I_{2|\mathcal{E}|}-H^* (I_{|\mathcal{E}|}\otimes C)) e^-[k]\nonumber\\&-(I_{|\mathcal{E}|}\otimes A)  H^* w[k]+ (V\otimes I_n) v[k].
     \end{align}
 Now, let us consider the tracking error, its error in estimation, and error in estimation of edge states of ensemble MACs, at the time of the first broadcast $s$,
    \begin{align}
        \tilde{\rho}[s]= \tilde{A} \tilde{\rho}[s-1]+ \tilde{\rho}_n[s-1],\label{rhotilde}
    \end{align}
    where $\tilde{\rho}[s]=\begin{bmatrix}\tilde{z}[s]^\top&e_{\tilde{z}}^-[s]^\top &e^-[s]^\top\end{bmatrix}^\top$ and 
    \begin{align}
       &\tilde{\rho}_n[s] =\nonumber\\
       &\begin{bmatrix}
           \bar{v}[s]- \bar{v}_G[s]\\
          -A H_{\tilde{z}}^* w_G[s] +\bar{v}[s]- \bar{v}_G[s]\\
       -(I_{|\mathcal{E}|}\otimes A)  H^* w[s]+ (V\otimes I_n) v[s]
       \end{bmatrix}.\nonumber
    \end{align}
    Iteratively expanding \eqref{rhotilde} gives
    \begin{align}
         \tilde{\rho}[s]= &\tilde{A} \tilde{A}_0^{(s-1)}\tilde{\rho}[0]+\sum_{i=2}^{s} \tilde{A} \tilde{A}_0^{(i-2)}\tilde{\rho}_n[s-i]\nonumber\\&+\tilde{\rho}_n[s-1].
    \end{align}
    This equation captures the evolution of the tracking error and its associated variables for duration $s$ that will be repeated periodically.
    Implementation of discrete time Lyapunov stability theorem dictates that if $\tilde{A} \tilde{A}_0^{(s-1)}$ is a Schur stable matrix then for a positive definite symmetric matrix $\tilde{Q}$ which is equal to variance of $\sum_{i=2}^{s} \tilde{A} \tilde{A}_0^{(i-2)}\tilde{\rho}_n[s-i]+\tilde{\rho}_n[s-1]$, there exists a positive definite symmetric matrix $\tilde{P}$ which is equal to variance of $\tilde{\rho}[s]$.
    As $\tilde{A} \tilde{A}_0^{(s-1)}$ is an upper triangular matrix, it is enough to ascertain Schur stability of its diagonal matrices, mainly, $ A_B(s-1)$ as matrices $A(I_2-H_{\tilde{z}}^* (-gC))$, $(I_{|\mathcal{E}|}\otimes A) (I_{2|\mathcal{E}|}- H_{\tilde{z}}^*(I_{|\mathcal{E}|}\otimes C))$ are stable, since they are based on Kalman filter algorithm. 
    Using the characteristic polynomial of matrix $A_B(s-1)$, its stability condition turns out to be $\frac{4}{(f_{b1}\tau s+2f_{b2})}>1$.
    This is the condition assigned in the optimisation problem. 
    As a result, the provided optimisation problem solves for $F_B$ such that $\lambda (A_B(s-1))\in \mathbb{D}$ is satisfied. 

Based on \eqref{ztildeexpansion} for the given setup, we can now conclude that for $k=ls, l\to \infty,$ we get$ \underset{k\to \infty}{\lim}\mathbb{E}[\tilde{z}[k]]=0$.
    This means that as $k\to \infty$, we have $\mathbb{E}[x_i[k]]=\mathbb{E}[\bar{z}[k]]=\mathbb{E}[\bar{Z}[k]],\;\forall\, i\in \mathcal{N}.$
  It is known that variance of error in estimation of tracking error, system noises of the ensemble clocks, and GNSS receiver clocks are finite, so we also conclude that $\underset{k\to \infty}{\lim}\mathbb{E}[\tilde{z}[k]\tilde{z}^\top[k]]$ is a finite matrix. All synchronised MACs together converge to the average of GNSS anchor clocks due to the influence of the tracking control that is broadcasted every $s$ interval. 
  This is a realistic scenario for generating synchronised time, which has optimal frequency stability in the long term, by a collection of MACs.
\end{proof}

\begin{rem}
    Eventhough the ensemble MACs undergo synchronisation consistently in the background, the synchronisation process does not directly interfere with the tracking process as seen in expansion \eqref{ztildeexpansion} and vice versa.
\end{rem}
\begin{rem}\label{rem2}
Although there are other ways of constructing the tracking control, we need to always check the possibility of interference of the synchronisation process.
For example, let us assume that, in place of \eqref{ug}, the tracking control is constructed as per the following equation, 
\begin{align}
    u_G'[k]=\frac{1}{g}(\boldsymbol{1}_n\boldsymbol{1}_g^\top\otimes F_B) \hat{\xi}_G^-[k],\label{ugedge}
\end{align}
where $\hat{\xi}_G^-[k]$ is the estimate of edge state $\xi_G[k]$. 
The total control input applied to the ensemble MACs in this case is
\begin{align}
    &u[k]=\begin{cases}
          u_{syn}[k]+u_G'[k]\hspace{-0.35cm}&,\text{if} \, k+1 = ls, l\in \mathbb{Z}^+, \\
          u_{syn}[k]\hspace{-0.35cm}&,\text{otherwise}.
    \end{cases}\label{ualt}
\end{align}
Then, the tracking error dynamics at the broadcast instant can be expanded as 
\begin{align}
    &\tilde{z}[ls]= A_B(s-1)^l\tilde{z}[0]\nonumber\\
  &  -\sum_{j=1}^l A_B(s-1)^{(l-j)}\bigg(\big(\frac{1}{g}q_A^\top \otimes B F_B\big)z[js-1]\nonumber\\
    &+\frac{1}{g}(\boldsymbol{1}^\top_g\otimes B F_B ) e_G^-[js-1]\bigg)+\sum_{j=1}^{l}A_B(s-1)^{(l-j)} \nonumber\\&\bigg(\sum_{i=0}^{(s-2)} A_B(i)  \bar{v}_T[js-2-i]+ \bar{v}_T[js-1]\bigg),\label{ztildelsedge2}
\end{align}
where $e_G^-[k]=\xi_G[k]-\hat{\xi}_G^-[k].$
Significant difference in the expression of $\tilde{z}[ls]$ in \eqref{ztildelsedge2} and \eqref{ztildeexpansion} lies in presence of synchronisation terms in the former which can disturb stabilisation of the core equation $\tilde{z}[ls]=A_B(s-1)^l \tilde{z}[0].$ 
This can in turn affect the AVAR of all the ensemble MACs as will be demonstrated in the following section.\end{rem}

\section{Numerical Results}\label{sec:numericalresults}

For the purpose of illustration, we provide some numerical results obtained from MATLAB R2024a. 
The proposed tracking framework according to Theorem \ref{thmfb} is applied to three synchronising MACs that interact as per the communication graph in Figure \ref{fig:example}. 
Of the three MACs, clock 1 and clock 3 are adjacent to GNSS receiver atomic clocks (GACs) 1 and 2, respectively. 
The control feedback gains $F$ is selected so that the condition in \eqref{fcondition} is satisfied and $F_B$ are designed according to \eqref{fbopti}. 
The assigned sampling interval $\tau$ is 1 Sec.  
The GNSS time tracking control input is broadcasted by the supervisor every 1000 Sec, \textit{i.e.} $s=1000$ Sec. 
The system noise variances for ensemble MACs are as per Table \ref{tab:variances of the noises}  and the noise variances of the measurements are as per Tables \ref{tab:variances of the measurement noises} and \ref{tab:variances of the measurement noises 2}.
\begin{figure}[t!]
  \centering
   \includegraphics[width=0.8\linewidth]{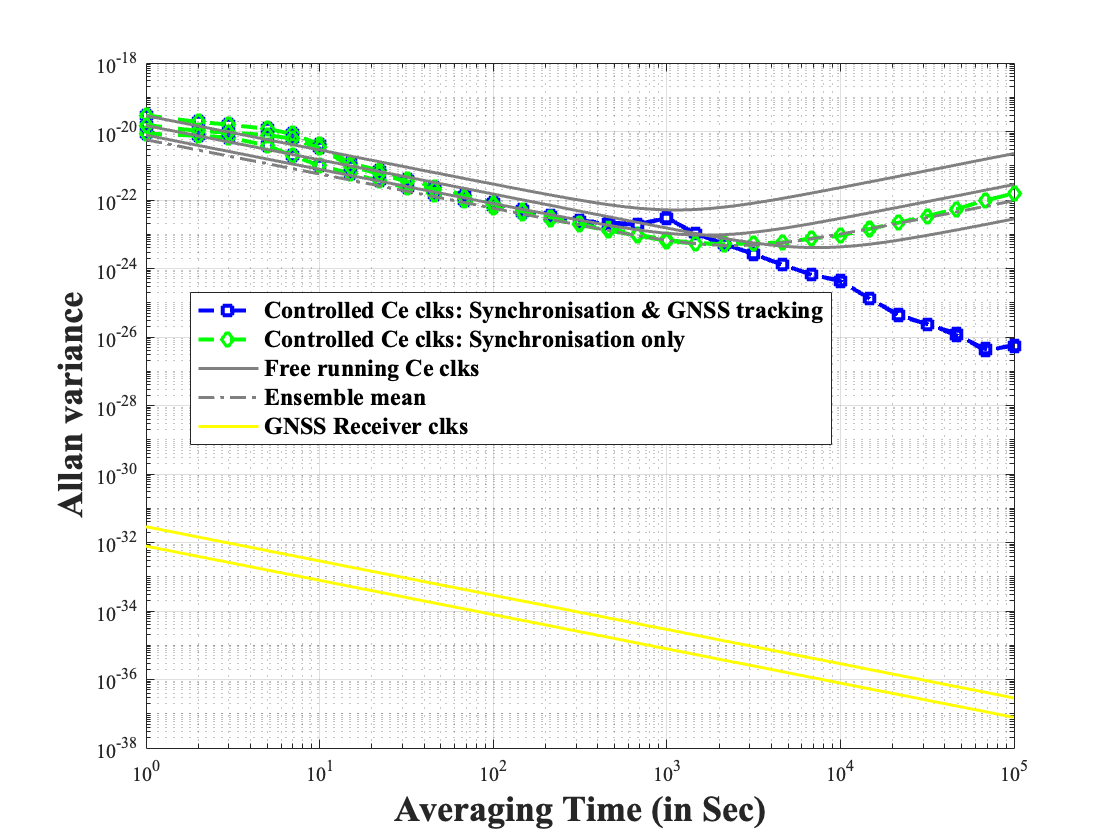}
    \caption{Statistical AVARs of the controlled MACs and the analytical AVARs of the free running MACs.}
    \label{fig:avsyn_track}
\end{figure}

\begin{figure}[t!]
    \centering
    \includegraphics[width=0.8\linewidth]{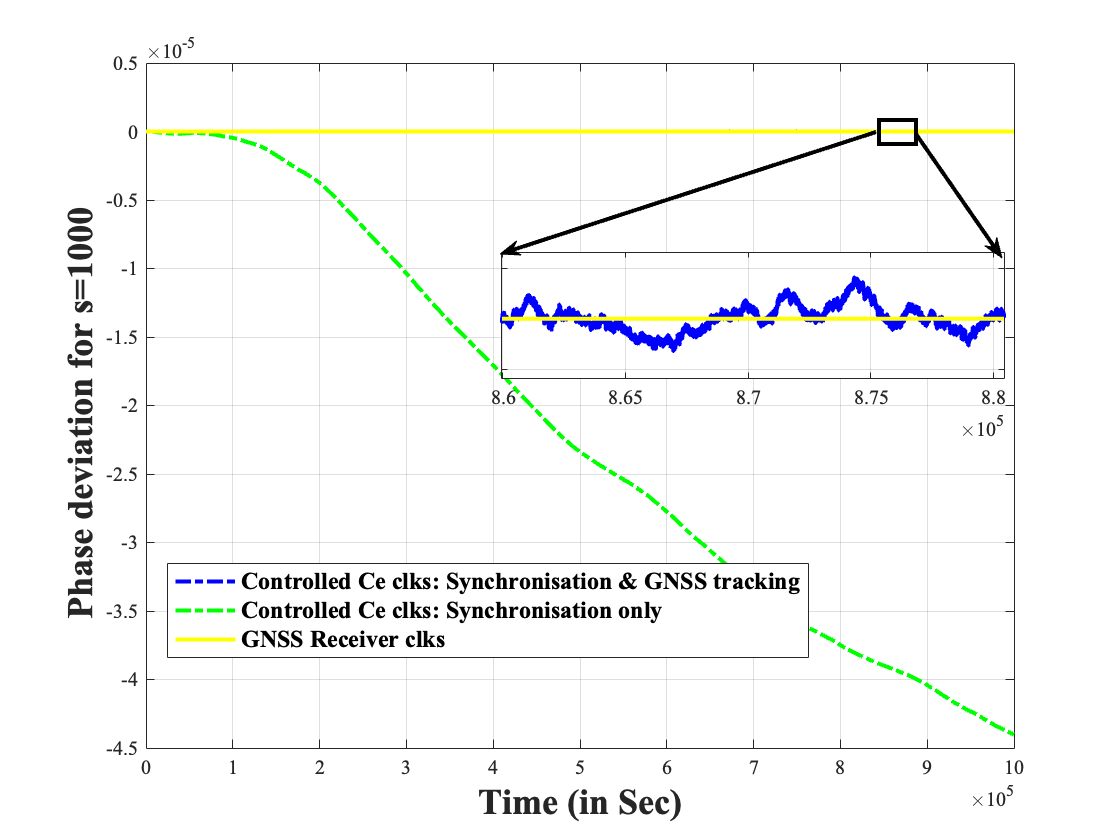}
    \caption{Time sequence of clock reading deviation of the MACs with respect to GACs.}
    \label{fig:crdsyn_track}
\end{figure}

\begin{table}[h!]
    \centering
    \begin{tabular}{c|c|c|c|c}
    Scale& $1\leftarrow 2$ & $1\rightarrow 2$& $2\leftarrow 3$& $2\rightarrow 3$\\
    \hline &&&&\\
      $10^{-28}$ &0.1895&0.0058&0.2228&0.0136\\&&&&
    \end{tabular}
    \caption{Variances of the noises in the edge measurements exchanged in the ensemble.}
    \label{tab:variances of the measurement noises}
\end{table}
\begin{table}[h!]
    \centering
    \begin{tabular}{c|c|c}
    Scale& GAC $1\rightarrow$ MAC $1$ & GAC $2\rightarrow$ MAC $3$\\
    \hline &&\\
      $10^{-16}$ &0.1721&0.0078\\&&
    \end{tabular}
    \caption{Variances of the noises in the edge measurements between GNSS anchor and their adjacent ensemble MACs }
    \label{tab:variances of the measurement noises 2}
\end{table}
In Figure \ref{fig:avsyn_track}, the three coinciding blue lines represents the AVARs of MACs when they are controlled as per \eqref{ufinal}, \textit{i.e.} they are synchronised to the virtual average clock, which in turn is steered towards the average of the GNSS receiver clocks periodically.
The three green plots represents the AVARs of MACs when they are controlled as per \eqref{usyni},  \textit{i.e.} they are synchronised to the virtual average clock only.  
The dotted gray line depicts the analytical mean AVAR of the MACS when they are free running, this has the best short term performance compared to any of the free running clocks.
Since, the green and blue plots gradually converge to the mean AVAR of MACs so they have good short term performance.
However, after $10^3$ Seconds, the performance of clocks that are only synchronising deteriorates.
On the other hand, due to the application of tracking control every $s=1000$ Seconds, the three blue plots further descend towards the analytical AVARs of the GACs (yellow plots). 
This indicates an improvement in the long-term frequency stability of the synchronising MACs.
This is reaffirmed by the clock reading deviation plots as shown in Figure \ref{fig:crdsyn_track}.
Here, three coinciding green plots represent the clock reading deviation of MACs that are only synchronising and they are randomly deviating away from the yellow lines which represent the clock reading deviation of the GACs. The zoomed subplot reveals the clock reading deviation of the MACs controlled according to \eqref{ufinal} which is represented by three coinciding blue plots oscillating about the yellow plots. 
\begin{figure}[t!]
    \centering
    \includegraphics[width=0.8\linewidth]{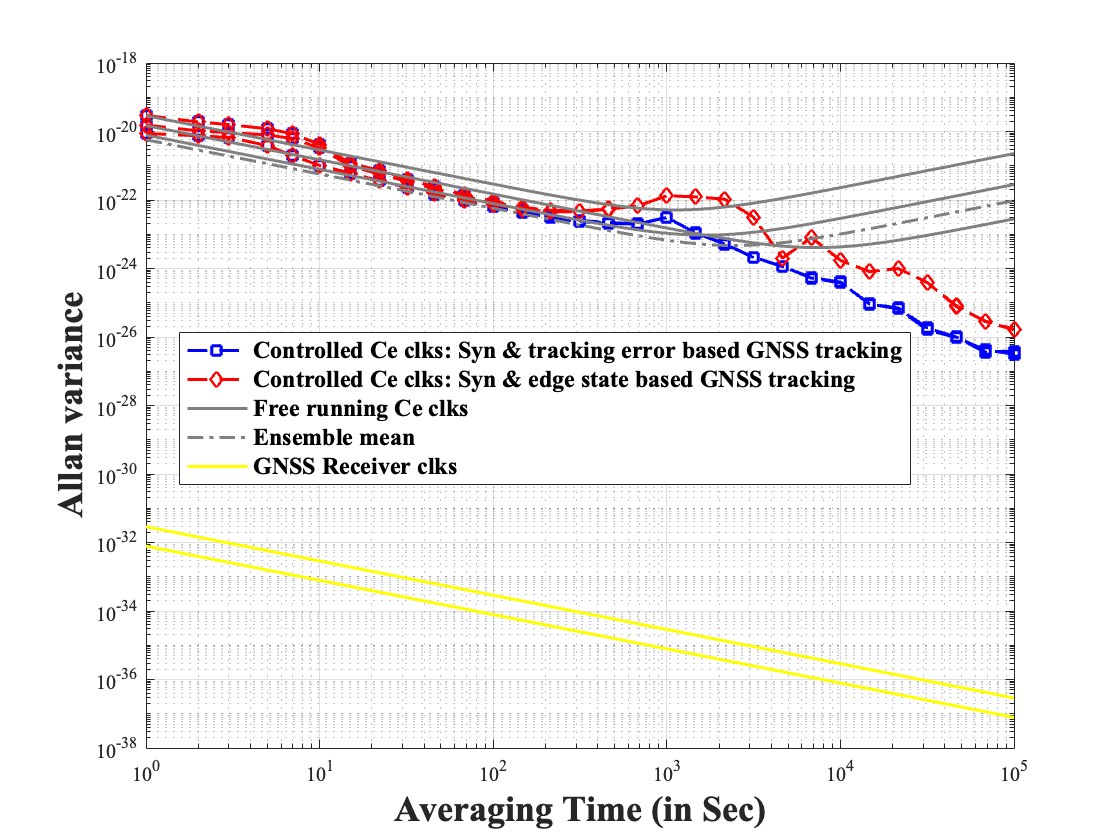}
    \caption{Comparison of statistical AVARs of the controlled MACs with the analytical AVARs of the free running MACs.}
    \label{fig:av_track_z_edge}
\end{figure}
\begin{figure}[t!]
    \centering
    \includegraphics[width=0.8\linewidth]{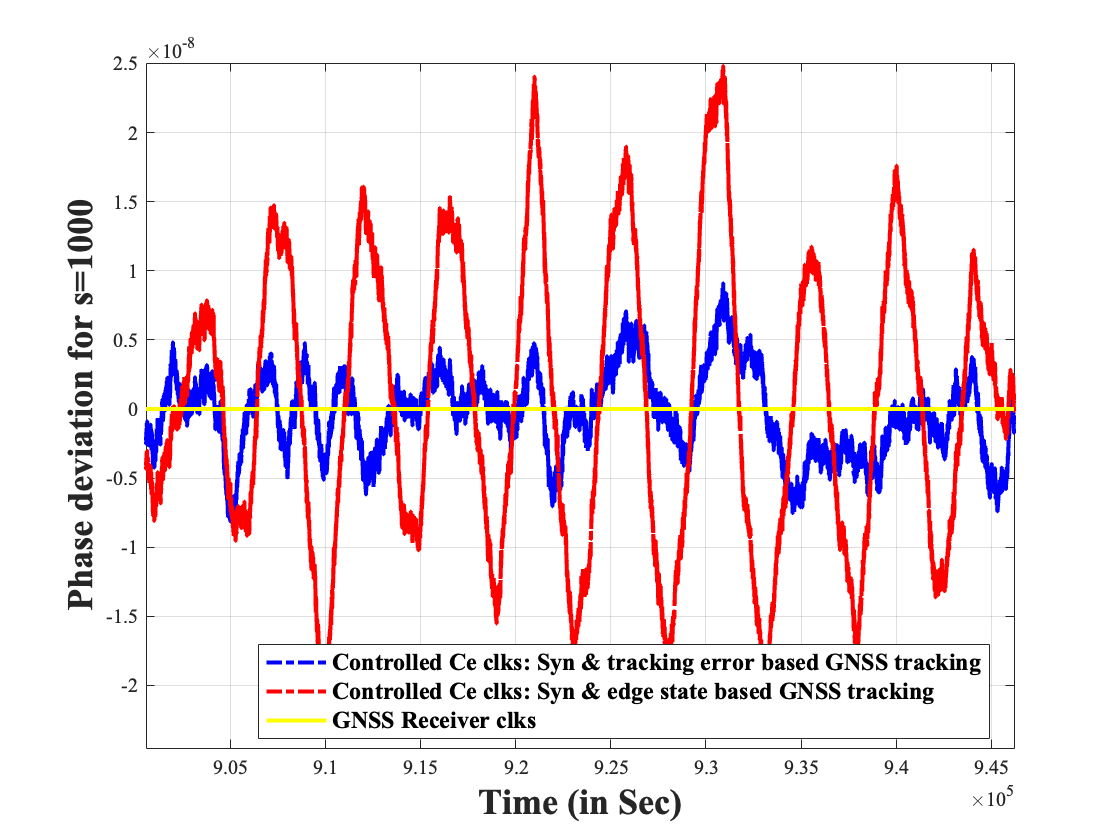}
    \caption{Time sequence of clock reading deviation of the MACs with respect to GACs.}
    \label{fig:cdr_track_z_edge}
\end{figure}

As explained in Remark \ref{rem2}, it is expected that the frequency stability of the MACs on application of feedback according to \eqref{ualt} will be inferior to the case when the feedback is as per \eqref{ufinal}. 
This is confirmed and demonstrated in Figure \ref{fig:av_track_z_edge}-\ref{fig:cdr_track_z_edge}. 
The three blue and red lines represent the performance of the MACs when the applied control signals are \eqref{ufinal} and \eqref{ualt}, respectively. 
It is imperative to note that the statistical AVAR is significantly better when the broadcasted tracking control is composed of the tracking error estimates as opposed to the GAC to adjacent clock edge estimates as there is a significant bump and subsequently upward shift in the red plots compared to the blue plots. 

In Figure \ref{fig:cdr_track_z_edge}, the oscillation amplitudes in the red plots are higher than in the blue plots. 
So, strategic application of broadcasted GNSS time tracking control to the synchronising MACs ensures that the synchronised time tracks the GNSS time to enhance its accuracy.
\section{Conclusion}\label{sec:conclusion}
In this manuscript, we propose an optimal GNSS time tracking algorithm for realisation of long-term stable and accurate time in an ensemble consisting of distributedly synchronising miniature atomic clocks. 
The proposed tracking control is constructed using the estimates of the tracking error which helps in stabilisation of the synchronised time generated by the MACs as it directly influences its ensemble mean. 
This results in reduction of AVAR over large averaging period leading to better long term performance by the MACs compared to the case where the MACs are only synchronising.
The current framework is tailored for an ensemble consisting of homogeneous MACs.
We are currently working towards extending the proposed tracking algorithm to the case of an ensemble consisting of heterogeneous MACs.

\end{document}